\documentclass[format=acmsmall, review=false]{acmart}

\usepackage{acm-ec-20}
\usepackage{booktabs} 

\newcommand{\be}{\begin{equation}}
\newcommand{\ee}{\end{equation}}

\begin{document}
\title[Heavy Tails Make Happy Buyers]{Heavy Tails Make Happy Buyers}
\author{Eric Bax}

\begin{abstract}
In a second-price auction with i.i.d. (independent identically distributed) bidder valuations, adding bidders increases expected buyer surplus if the distribution of valuations has a sufficiently heavy right tail. While this does not imply that a bidder in an auction should prefer for more bidders to join the auction, it does imply that a bidder should prefer it in exchange for the bidder being allowed to participate in more auctions. Also, for a heavy-tailed valuation distribution, marginal expected seller revenue per added bidder remains strong even when there are already many bidders. 
\end{abstract}

\maketitle

\section{Introduction} \label{sec_introduction}
In a second-price sealed-bid auction, also known as a Vickrey auction \cite{vickrey61}, bidders submit their bids, the item is awarded to the bidder with the highest bid, and the winner pays the runner-up bid. The seller's revenue is the runner-up bid. Adding a bidder increases expected revenue for the seller by increasing competition: if the new bid is the highest bid, then the new bidder wins and pays the previous highest bid; if the new bid is the runner-up, then it increases the price for the winner.  

In this paper, we consider whether adding a bidder can also increase expected buyer surplus. Assuming bidders bid their valuations, which is a dominant strategy \cite{vickrey61,krishna02,milgrom05}, the buyer's surplus is the difference between the highest bid (the buyer's valuation) and the runner-up bid (the buyer's price). If the new bid is the highest bid, then the new bidder becomes the new buyer and buyer surplus increases if the difference between the new bid and the previous highest bid is greater than the difference between the previous highest and runner-up bids. So buyer surplus benefits from valuations that are right-tail outliers. We show that exponential distributions maintain expected buyer surplus as bidders are added, and Pareto distributions increase expected buyer surplus. We also show that as tails get heavier the ratio of expected buyer surplus to expected seller revenue increases, even with optimal reserve prices.

Ideally, we would show that adding a bidder is a Pareto improvement \cite{pareto96} (in expectation), meaning that it benefits the seller or one of the bidders and it harms no seller or bidder. But this is impossible: adding a bidder decreases expected surplus for the previous bidders, because the new bid may reduce the previous winner's surplus by being the new runner-up bid or eliminate the previous winner's surplus by being the new winning bid. However, we show that adding a bidder is a Pareto improvement in expectation if the additional bidder is symmetric to the previous bidders. For example, if adding a bidder increases expected buyer surplus, then it is a Pareto improvement to have all bidders participate in an auction rather than select a subset of bidders at random and exclude the others. Similarly, for a set of auctions with each bidder excluded from one auction and all others participating, it is a Pareto improvement for all bidders to participate in all auctions.

Our analysis uses order statistics \cite{david03} of the distribution of bidder valuations. Order statistics are often used to analyze auctions \cite{krishna02,milgrom05,nisan07}, since, for truthful second-price auctions, the winning bid is the greatest order statistic among valuations and the price is the next-greatest order statistic. Second-price auctions (and variations on them) are the basis for online marketplaces for advertising \cite{edelman07,varian09} and for goods and services \cite{milgrom05}. Distributions of valuations are a subject of much study, because they determine optimal reserve prices \cite{myerson81,riley81}, and estimates of the distributions determine reserve prices for actual marketplaces \cite{ostrovsky09,cole14,blum15}. 

Our results depend on the shape of the valuation distribution. With a uniform distribution over some range, adding more bidders drives both the average winning bid and the average runner-up bid toward the right of the range, pushing these averages closer to each other, which squeezes the winner's surplus. In contrast, with a heavy-tailed distribution, adding bidders makes it more likely that an extreme right outlier bid will occur, boosting the winner's surplus as well as the seller's revenue. Valuation distribution shapes play a role in some other results about auctions, including determining whether imposing Pigovian taxes \cite{pigou12} on bidders to account for their negative externalities \cite{burrows79,eskeland92,abrams08} can increase surplus for bidders as a class \cite{stourm17}. Heavy tails are also a topic of interest in the study of complexity theory \cite{waldorp92,carlson02,hilbert14,wierman14}, in marketing \cite{anderson06}, in characterizing distributions based on samples \cite{kulldorff73,vannman76,markovich07}, and in the study of a variety of physical and social phenomena \cite{pareto96,seal80,reed04,schroeder10}.

\section{Heavy Tails and Buyer Surplus}
Assume there are $n+1$ bidders. Let $X_1, \ldots, X_{n+1}$ be the bidders' valuations, drawn i.i.d. from a distribution $D$. Let $X_{(1)}, \ldots, X_{(n+1)}$ be the order statistics for $X_1, \ldots, X_{n+1}$ -- the values ranked from least to greatest, with any ties broken randomly. 

Let $p_{n+1}$ be the buyer's surplus in a second-price auction if all $n+1$ bidders participate, and let $p_n$ be the buyer's surplus if only the first $n$ bidders participate. Assume each bidder bids their private value, since the second-price auction is truthful \cite{vickrey61,krishna02,milgrom05}. Recall that buyer surplus is the difference between the highest bid and the runner-up bid. The following theorem gives a condition for adding bidder $n+1$ to increase expected buyer surplus, based on the order statistics of drawing $n+1$ bids i.i.d. from $D$.

\begin{theorem} \label{thm1}
$$ E(p_{n+1} - p_n) = \frac{1}{n+1} [E(X_{(n+1)} - X_{(n)}) - 2 E(X_{(n)} - X_{(n-1)})], $$
where all expectations are over $(X_1, \ldots, X_{n+1}) \sim D^{n+1}$.
\end{theorem}

\begin{proof}
To assess the impact of adding bidder $n+1$ to an auction with $n$ bidders, we will assess the (opposite) impact of removing bidder $n+1$ from an auction with $n+1$ bidders. Note that
$$ p_{n+1} = X_{(n+1)} - X_{(n)}. $$

With probability $\frac{1}{n+1}$, $X_{n+1}$ is $X_{(n+1)}$ (the highest bid), so removing bidder $n+1$ from the auction removes the winning bid. Then the new winning bid is $X_{(n)}$, and the new runner-up bid is $X_{(n-1)}$. So buyer surplus becomes
$$ p_n = X_{(n)} - X_{(n-1)}. $$
Also with probability $\frac{1}{n+1}$, $X_{n+1}$ is $X_{(n)}$. If so, then removing bidder $n+1$ removes the runner-up bid, decreasing the second price from $X_{(n)}$ to $X_{(n-1)}$ and increasing buyer surplus to 
$$ p_n = X_{(n+1)} - X_{(n-1)}. $$
Alternatively, with probability $\frac{n-1}{n+1}$, $X_{n+1}$ is neither the highest nor the runner-up bid. Then removing bidder $n+1$ has no effect on buyer surplus: $p_n = p_{n+1}$.

The cases are: (a) $X_{n+1}$ is $X_{(n+1)}$, (b) $X_{n+1}$ is $X_{(n)}$, and (c) $X_{n+1}$ is one of $X_{(1)}, \ldots, X_{(n-1)}$. For each case, multiply the probability of occurrence by the expectation of $p_{n+1} - p_n$ given the occurrence:
$$ E (p_{n+1} - p_n) $$
$$ = \frac{1}{n+1} E [(X_{(n+1)} - X_{(n)}) - (X_{(n)} - X_{(n-1)})]$$
$$ + \frac{1}{n+1} E [(X_{(n+1)} - X_{(n)}) - (X_{(n+1)} - X_{(n-1)})] $$
$$ + \frac{n-1}{n+1} E [0]. $$
Use linearity of expectation and gather terms:
$$ = \frac{1}{n+1} E [X_{(n+1)} - 3 X_{(n)} + 2 X_{(n-1)}]. $$
Use linearity of expectation again and rearrange terms:
$$ = \frac{1}{n+1} [E(X_{(n+1)} - X_{(n)}) - 2 E(X_{(n)} - X_{(n-1)})].$$
\end{proof}

The theorem says that adding a bidder increases expected buyer surplus if the expected difference between the top two bids with $n+1$ bidders is at least twice the expected difference between the second and third-highest bids. For this to occur, the bid distribution must tend to supply bids with the first gap between bids at least twice the second gap. In other words, $D$ must tend to produce winners that are outliers.

\subsection{The Border: Exponential}
Exponential distributions have cdf:
\[ F(x) = \left\{ 
\begin{array}{ll}
1 - e^{-\lambda x} & \mbox{if $x \geq 0$}\\
0 & \mbox{if $x < 0$}\end{array} 
\right., \]
where $\lambda > 0$. 

For exponential distributions, adding a bidder does not affect the expected buyer surplus in a second-price auction:

\begin{corollary}
If $D$ is an exponential distribution, then $E (p_{n+1} - p_n) = 0$.
\end{corollary}

\begin{proof}
We will show that 
$$ E(X_{(n+1)} - X_{(n)}) - 2 E(X_{(n)} - X_{(n-1)}) = 0. $$
The order statistics for an exponential distribution have known distributions \cite{renyi53} \cite{feller66} (pg. 20) \cite{david03} (Ch. 2):
$$ X_{(i)} \sim \frac{1}{\lambda} \left( \sum_{j=1}^{i} \frac{Z_j}{n-j+2} \right), $$
where $Z_j$ are exponential random variables with $\lambda = 1$. Since the mean of an exponential distribution is $\frac{1}{\lambda}$, $\forall j: E Z_j = 1$. So
$$ E(X_{(n+1)} - X_{(n)}) - 2 E(X_{(n)} - X_{(n-1)}) $$
$$ = \left(\frac{1}{\lambda} \sum_{j=1}^{n+1} \frac{1}{n-j+2} - \frac{1}{\lambda} \sum_{j=1}^{n} \frac{1}{n-j+2} \right) $$
$$- 2 \left(\frac{1}{\lambda} \sum_{j=1}^{n} \frac{1}{n-j+2} - \frac{1}{\lambda} \sum_{j=1}^{n-1} \frac{1}{n-j+2} \right) $$
$$ = \frac{1}{\lambda} \left(\sum_{j=n+1}^{n+1} \frac{1}{n-j+2} - 2 \sum_{j=n}^{n} \frac{1}{n-j+2}\right) $$
$$ = \frac{1}{\lambda} \left(1 - 2 (\frac{1}{2}) \right) $$
$$ = 0.$$ 
\end{proof}

Compare expected buyer surplus to expected seller revenue:

\begin{theorem} \label{thm_ratio_exp}
For second-price auctions with $n+1$ bidders having i.i.d. exponential valuation distributions, the ratio of expected seller revenue to expected buyer surplus is approximately $\ln (n+1) + \gamma - 1$, where $\gamma$ is the Euler - Mascheroni constant ($\gamma \approx 0.5772$).
\end{theorem}

\begin{proof}
Expected buyer surplus is the expected difference between the top two bids:
$$ E(X_{(n+1)} - X_{(n)}) $$
$$ = \left(\frac{1}{\lambda} \sum_{j=1}^{n+1} \frac{1}{n-j+2} - \frac{1}{\lambda} \sum_{j=1}^{n} \frac{1}{n-j+2} \right) $$
$$ = \frac{1}{\lambda}. $$
Expected seller revenue is the expected second bid:
$$ E X_{(n)} = \frac{1}{\lambda} \sum_{j=1}^{n} \frac{1}{n-j+2} $$
$$ = \frac{1}{\lambda} \left(\frac{1}{2} + \frac{1}{3} + \ldots + \frac{1}{n+1} \right) $$
$$ \approx \frac{1}{\lambda} \left(\ln (n+1) + \gamma - 1 \right) $$
because $1 + \frac{1}{2} + \ldots + \frac{1}{n} \approx \ln n + \gamma$.
\end{proof}

The seller takes the lion's share of the total auction value as the number of bidders increases: as $n \rightarrow \infty$, $\ln n \rightarrow \infty$, but expected buyer surplus remains constant. Can the buyer instead take the lion's share? Yes, if the tail is heavy enough, as we show later. 

But first, consider how using an optimal auction \cite{myerson81,riley81} with a reserve price, instead of a straight second-price auction, would affect our results. In Appendix \ref{optimal_expo}, we show that with the Myerson-optimal reserve price of $\frac{1}{\lambda}$, the expected buyer surplus with $n+1$ bidders is 
$$  \frac{1}{\lambda} [1 - (1 - e^{-1})^{n+1}]. $$
This converges from below to the result for the second-price auction, $\frac{1}{\lambda}$, as $n$ increases, because the reserve price is less likely to have an impact as more bidders are added. So expected buyer surplus increases with each new bidder, but the amount of increase decreases with each new bidder, making the increase in expected buyer surplus nearly zero for large $n$. 

Both auction types have the same sum of seller revenue and buyer surplus; it is the top bid. Since the optimal auction charges the buyer a bit more in expectation, it increases expected seller revenue by the same amount: $\frac{1}{\lambda} (1 - e^{-1})^{n+1}$, and expected seller revenue still dominates expected buyer surplus as $n$ increases.

\subsection{Over the Border: Pareto}
Type I Pareto distributions have cdf:
\[ F(x) = \left\{ \begin{array}{ll}
1 -  \left(\frac{a}{x}\right)^{v} & \mbox{if $x \geq a$}\\
0 & \mbox{if $x < a$} \end{array} \right. , \]
where $a > 0$ and $v > 1$.

Pareto distributions were developed to model the distribution of wealth in societies \cite{pareto96}. Since then, they have been used to model many distributions, including those of stock price returns \cite{reed04}, hard drive errors \cite{schroeder10}, reserves per oil field \cite{reed04}, and populations of human settlements \cite{reed04}. Pareto distributions are heavy-tailed, meaning that they have a higher probability of a maximum sample that is much greater than the other samples, compared to exponential distributions. 

For $n+1$ samples, the expectations of order statistics are \cite{huang75} \cite{david03} (pg. 52):
$$ E X_{(i)} = a \frac{(n+1)!}{(n+1-i)!} \frac{\Gamma(n + 2 - i - \frac{1}{v})}{\Gamma(n + 2 - \frac{1}{v})},$$
where $\Gamma()$ is the gamma function, which is a continuous version of the factorial: $\Gamma(n + 1) = n!$ for all integers $n$ and $\Gamma(z+1) = z \Gamma(z)$ for all real $z > 0$. We use both of these equalities to simplify expectations for order statistics in the next proof.

\begin{corollary}
If $D$ is a Type I Pareto distribution, then adding a bidder increases expected buyer surplus: $E (p_{n+1} - p_n) > 0$.
\end{corollary}

\begin{proof}
We will show that 
$$ E(X_{(n+1)} - X_{(n)}) - 2 E(X_{(n)} - X_{(n-1)}) > 0. $$
Let 
$$ g(n, a, v) \equiv a \frac{\Gamma(n+2)}{\Gamma(n+2-\frac{1}{v})} \Gamma(1 - \frac{1}{v}). $$
Then
$$ E X_{(n+1)} = g(n, a, v), $$
$$ E X_{(n)} = g(n, a, v) (1 - \frac{1}{v}),$$
and
$$ E X_{(n-1)} = \frac{1}{2} g(n, a, v) (1 - \frac{1}{v}) (2 - \frac{1}{v}).$$
So
$$ E(X_{(n+1)} - X_{(n)}) - 2 E(X_{(n)} - X_{(n-1)}) $$
$$ = g(n, a, v) \left( \left[1 - (1 - \frac{1}{v})\right] - 2 \left[ (1 - \frac{1}{v}) - \frac{1}{2} (1 - \frac{1}{v}) (2 - \frac{1}{v})\right] \right) $$
The terms in $g()$ are all positive. The term in parentheses is:
$$ \frac{1}{v} - (1 - \frac{1}{v})(2 - (2 - \frac{1}{v}))  $$
$$ = \frac{1}{v} - (1 - \frac{1}{v}) \frac{1}{v} $$
$$ = \frac{1}{v^2} > 0. $$ 
\end{proof}

Compare expected buyer surplus to expected seller revenue. Expected buyer surplus is
$$ E X_{(n+1)} - E X_{(n)} $$
$$ = g(n, a, v) \left(1 - (1 - \frac{1}{v})\right) $$
$$ =  g(n, a, v) \frac{1}{v}. $$
Expected seller revenue is
$$ E X_{(n)} = g(n, a, v) (1 - \frac{1}{v}). $$

Since 
$$ \frac{\Gamma(n+2)}{\Gamma(n+2-\frac{1}{v})} \approx (n+1)^{\frac{1}{v}}, $$
expected buyer surplus is
$$ \approx a \frac{1}{v} (n+1)^{\frac{1}{v}} \Gamma(1 - \frac{1}{v}),$$
and expected seller revenue is 
$$ \approx a (1 - \frac{1}{v}) (n+1)^{\frac{1}{v}} \Gamma(1 - \frac{1}{v}).$$
Both increase with the number of bidders.

The ratio of expected seller revenue to expected buyer surplus is:
$$ \frac{1 - \frac{1}{v}}{\frac{1}{v}} = v - 1.$$
So expected seller revenue equals expected buyer surplus at $v = 2$. As $v \rightarrow \infty$, the seller takes more of the value; as $v \rightarrow 1$, the buyer takes more.

\subsection{Beyond, to Infinity}
Consider what happens as $v \rightarrow 1$ in more detail. Recall that expected seller revenue is 
$$ E X_{(n)} = g(n,a,v) (1 - \frac{1}{v})$$
$$ = a \frac{\Gamma(n+2)}{\Gamma(n+2-\frac{1}{v})} \Gamma(1 - \frac{1}{v}) (1 - \frac{1}{v}),$$
and $\Gamma(1 - \frac{1}{v}) (1 - \frac{1}{v}) = \Gamma(2 - \frac{1}{v})$, so
$$ = a \frac{\Gamma(n+2)}{\Gamma(n+2-\frac{1}{v})} \Gamma(2 - \frac{1}{v}).$$
Since $\Gamma(n+1) = n!$ for all integers $n$ and $\Gamma(1) = 1$, expected seller revenue approaches $a (n+1)$.

For expected buyer surplus:
$$ \lim_{v \rightarrow 1} \frac{1}{v} g(n, a, v) $$
$$ =  \lim_{v \rightarrow 1} \frac{1}{v} a \frac{\Gamma(n+2)}{\Gamma(n+2-\frac{1}{v})} \Gamma(1 - \frac{1}{v}) $$
$$ = a (n + 1) \lim_{v \rightarrow 1} \Gamma(1 - \frac{1}{v}).$$
And
$$ \lim_{z \rightarrow 0^{+}} \Gamma(z) = \infty. $$
So expected buyer surplus approaches infinity. 

Can the seller impose a reserve price to raise expected seller revenue? The answer is no, because with $v = 1$ we have the equal revenue distribution. To see this, consider expected seller revenue with reserve price $r$ and $v = 1$ for $n$ bidders:
$$ n \int_{r}^{\infty} r F(r)^{n-1} f(x) \, dx $$
$$ + n (n-1) \int_{r}^{\infty} x F(x)^{n-2} f(x) [1 - F(x)] \, dx $$
The first term is revenue from the reserve price: $n$ possible top bidders, with revenue $r$ if the other $n-1$ bids are below $r$ (probability $F(r)^{n-1}$) and the top bid is above $r$ (probability $f(x) \, dx$ integrated over $x$ from $r$ to $\infty$). The second term is revenue from the second price: $n (n-1)$ possible top and second bidder combinations, with revenue the second price $x$ if the other $n-2$ bids are below $x$ (probability $F(x)^{n-2}$) and the top bid is above $x$ (probability $[1 - F(x)]$). 

The first term (revenue from reserve price) is 
$$ = anr(1 - \frac{a}{r})^{n-1} \int_{r}^{\infty} x^{-2} \, dx $$
$$ = anr (1 - \frac{a}{r})^{n-1} \left. \left[- x^{-1}\right] \right|_{r}^{\infty} $$
$$ = an (1 - \frac{a}{r})^{n-1}.$$
The second term (revenue from second price) is
$$ = n (n-1) \int_{r}^{\infty} a^2 (1 - \frac{a}{x})^{n-2} \, dx $$
$$ = an \left. \left[ (1 - \frac{a}{x})^{n-1} \right] \right|_{r}^{\infty} $$
$$ = an - an (1 - \frac{a}{r})^{n-1}.$$
The two terms sum to $an$; the reserve price has no effect. For this reason, the distribution with $v=1$ is known as the equal-revenue distribution. (For more discussion of optimal reserve prices and Type I Pareto valuations, refer to Appendix \ref{optimal_expo}.)

\section{Should Bidders Welcome New Bidders?}
If adding a bidder increases expected surplus for the winner, then should $n$ bidders welcome one more bidder to an auction? No, because the new bidder can decrease surplus for the previous winner -- by being the new winner or the new runner-up -- and cannot increase surplus for the previous bidders. Stated another way: adding a bidder decreases other bidders' probabilities of winning and may potentially raise the winner's price. In fact:

\begin{theorem} \label{per_bidder_thm}
For any bidder valuation distribution $D$ that has support over at least two different values, expected surplus per participating bidder decreases with each additional bidder:
$$ \frac{E p_{n+1}}{n+1} < \frac{E p_n}{n}. $$
\end{theorem}

\begin{proof}
Because the new bidder can take the win from the previous winner or increase their price, adding a bidder reduces expected surplus for the previous bidders. By symmetry among bidders, the new bidder has the same expected surplus as the other bidders. So the added bidder must decrease expected surplus per participating bidder. 

This argument fails only if the top two bids among the previous bidders are equal with probability one. In that case, there is zero expected surplus, and adding a new bidder does not change that. But the theorem requires $D$ to have support over at least two different values, so the probability of a tie among the top two bids is less than one.
\end{proof}

If we consider the new bidder and previous bidders as a class, then the expected surplus per class member increases with wider participation if $E p_{n+1} > E p_n$, because that clearly implies:
$$ \frac{E p_{n+1}}{n+1} > \frac{E p_n}{n+1}. $$
Adding bidder $n+1$ increases their expected surplus from zero to $\frac{E p_{n+1}}{n+1}$. But it changes each previous bidder's expected surplus from $\frac{E p_n}{n}$ to $\frac{E p_{n+1}}{n+1}$, which is a decrease according to Theorem \ref{per_bidder_thm}. So this is not a Pareto improvement, because it harms the previous bidders. 

However, there is a Pareto improvement if each of the $n+1$ potential bidders are symmetric. For example, if one bidder selected at random is to be excluded from the auction, then it is a Pareto improvement to instead have all potential bidders participate, because each potential bidder increases their expected surplus by:
$$\frac{E p_{n+1}}{n+1} - \frac{E p_n}{n+1}. $$
Similarly, suppose there is a set of $n+1$ auctions, each auction excludes one bidder, and each bidder is excluded from one auction. Then it is a Pareto improvement to allow all bidders to participate in all auctions, because each of the $n+1$ bidders increases their expected surplus by:
$$(n+1) \frac{E p_{n+1}}{n+1} - n \frac{E p_n}{n}$$
$$ = E p_{n+1} - E p_n.$$

For a more general framework to evaluate whether a change in auction participation over multiple auctions is a Pareto improvement in expectation, assume that all bidders draw their valuations i.i.d. for each auction, though there may be different valuation distributions for different auctions. Let $i$ index potential bidders. Let $j$ index auctions. Let $I(i,j)$ be one if bidder $i$ is allowed to bid in auction $j$ and zero otherwise. For each auction $j$, let
$$ b_j = \sum_i I(i,j) $$
be the number of bidders. Let $E p_j(b)$ be the expected buyer surplus for auction $j$ if $b$ bidders participate. Then the expected surplus for potential bidder $i$ over all the auctions is
$$ \sum_j I(i,j) \frac{E p_j(b_j)}{b_j}. $$
Let $E r_j(b)$ be the seller's expected revenue for auction $j$ if $b$ bidders participate. If there is a single seller for all auctions, then their expected revenue over all auctions is
$$ \sum_j E r_j(b_j). $$
For a new bidding participation arrangement, let $I'(i,j)$ indicate whether bidder $i$ is allowed to bid in auction $j$. For each auction $j$, let
$$ b_j' = \sum_i I'(i,j) $$
be the new number of bidders. Then the new arrangement is a Pareto improvement in expectation if 
$$ \forall i: \sum_j I'(i,j) \frac{E p_j(b_j')}{b_j'} \geq \sum_j I(i,j) \frac{E p_j(b_j)}{b_j},$$
and, if there is a single seller,
$$ \sum_j E r_j(b'_j) \geq \sum_j E r_j(b_j), $$
and the inequality is strict for at least one bidder or the seller. If each auction has a different seller, then we require
$$ \forall j: E r_j(b'_j) \geq E r_j(b_j) $$
instead of the previous inequality.

\section{Heavy Tails and Seller Revenue}
Consider how the number of bidders affects expected seller revenue. In accord with common sense, having more bidders increases competition, which increases expected seller revenue:

\begin{theorem}
For any valuation distribution $D$ that has support over multiple values, adding a bidder to a second-price auction increases expected seller revenue.
\end{theorem}

\begin{proof}
Seller revenue is the second-highest bidder valuation. Regardless of the rank of the additional bid among itself and the previous bids, the second-highest bid cannot decrease. Let $s$ be the second-highest bid, excluding the additional one. Since there is support over multiple values, there is positive probability that $s$ is not a maximum among the distribution values that have support. If it is not a maximum, then there is positive probability that the additional bid is greater than $s$, which increases the second-highest bid if we include the additional bid.
\end{proof}

In general, heavier-tailed valuation distributions produce more expected seller revenue per added bidder:

\begin{theorem}
Let $s_{n+1}$ be expected seller revenue for $n+1$ bidders and let $s_n$ be expected seller revenue for $n$ bidders. Then:
\begin{itemize}
\item For a uniform valuation distribution: $D \sim U[0,1]$, $s_{n+1} - s_n \in \hbox{O}\left(\frac{1}{n^2}\right).$
\item For an exponential valuation distribution, $s_{n+1} - s_n \in \hbox{O}\left(\frac{1}{n}\right). $
\item For a Type I Pareto valuation distribution, $s_{n+1} - s_n \in \hbox{O}\left(\frac{1}{n^{1-\frac{1}{v}}}\right). $
\end{itemize} 
\end{theorem}

\begin{proof}
To examine the expected increase in expected seller revenue from adding a bidder to a competition with $n$ bidders, analyze the difference in expected second-highest value over $n+1$ and over $n$ values for the valuation distributions. For the (non-heavy-tailed) uniform distribution, the expectation of the $i$th order statistic for $n$ samples is $\frac{i}{n+1}$ \cite{david03}. So the expected increase in seller revenue is 
$$ s_{n+1} - s_n = \frac{n}{n + 2} - \frac{n - 1}{n + 1} = \frac{2}{n^2 + 3 n + 2} \in \hbox{O}\left(\frac{1}{n^2}\right).$$

Next, consider the exponential distribution. Recall that
$$ s_{n+1} = \frac{1}{\lambda} \left(\frac{1}{2} + \frac{1}{3} + \ldots + \frac{1}{n+1} \right). $$
So 
$$ s_{n+1} - s_n = \frac{1}{\lambda} \left(\frac{1}{n+1}\right) \in \hbox{O}\left(\frac{1}{n}\right). $$

For the Type I Pareto distribution, 
$$ s_{n+1} \approx a (1 - \frac{1}{v}) (n+1)^{\frac{1}{v}} \Gamma(1 - \frac{1}{v}).$$
So
$$ s_{n+1} - s_n \approx a (1 - \frac{1}{v}) \Gamma(1 - \frac{1}{v}) \left[(n+1)^{\frac{1}{v}} - n^{\frac{1}{v}}\right]. $$
Since 
$$ \frac{\partial}{\partial n} n^{\frac{1}{v}} = \frac{1}{v} \frac{1}{n^{1-\frac{1}{v}}}, $$
$$ s_{n+1} - s_n \in \hbox{O}\left(\frac{1}{n^{1-\frac{1}{v}}}\right). $$
\end{proof}

For the Pareto Type I distribution, as $v \rightarrow \infty$, the limit is O$\left(\frac{1}{n}\right)$, as for the exponential distribution. For $v = 2$, it is O$\left(\frac{1}{\sqrt{n}}\right)$. As $v \rightarrow 1_{+}$, the limit is O$(1)$, indicating that as the tail becomes heavier, the seller gains nearly as much from adding a bidder when there are already many bidders as from adding a bidder when there are only a few.

\section{Conclusion}
We have shown that expected buyer surplus can increase with increased competition if the valuation distribution is sufficiently heavy-tailed. While adding a new bidder can never increase per-bidder expected surplus, it can increase expected surplus over the class of new and previous bidders. As a result, under sufficiently heavy-tailed valuation distributions, each potential bidder should prefer wider participation in auctions for themselves and their fellow bidders over narrower participation for all, if the restrictions on their participation are symmetric to those on other bidders. 

In the future, it would be interesting to explore whether or how heavy-tailed valuation distributions occur in practice. They seem most likely to occur in scenarios in which great gains are possible. For example, if we model approaches to curing serious and widespread health conditions (such as cancer or aging) as bids and the value to society for solutions as valuations, then a heavy-tailed valuation distribution corresponds to the possibility that some creative approach will yield great progress. In such cases, adding efforts to develop creative solutions increases the expected profit for the provider of the best solution (corresponding to buyer surplus if the provider can charge as much as the value of the next-best solution), in addition to increasing the expected benefit to society. 

In some cases, valuation distributions may mimic heavy-tailed distributions over some range, but lack support above some finite upper bound. In these cases, it would be interesting to analyze marginal expected buyer surplus per added bidder as a function of the number of bidders. If the valuation distribution is unknown, but there is bid information from previous auctions with the same distribution, it would be interesting to explore whether it is possible to accurately estimate the number of bidders needed to maximize expected buyer surplus. It would also be interesting to examine how non-i.i.d. valuation distributions affect the results in this paper, especially if some potential bidders have heavy-tailed valuation distributions and others do not. 


\appendix

\section{Heavy Tails and Reserve Prices} \label{optimal_expo}
In this appendix, we examine how using a reserve price to convert the second price auction to the Myerson \cite{myerson81} optimal auction affects our results in this paper. For the exponential distribution, a reserve price decreases expected buyer surplus, but the difference between second-price and Myerson-optimal expected buyer surplus decreases exponentially in the number of bidders. Interestingly, for the Pareto Type I distribution, a reserve price does not make sense, because no reserve price increases expected seller revenue to more than that from a straight second-price auction. 

\subsection{Reserves for Exponential}
First, observe that the exponential distribution meets the required Myerson regularity condition: that
$$ x - \frac{1 - F(x)}{f(x)}  = x - \frac{1}{\lambda}$$
is monotonically increasing in $x$. Set this equal to zero and solve for the optimal reserve price: 
$$r = \frac{1}{\lambda}.$$

\begin{theorem}
For $n$ bidders with valuations drawn i.i.d. from an exponential distribution, optimal-auction expected buyer surplus is $\frac{1}{\lambda} [1 - (1 - e^{-1})^{n}].$
\end{theorem}

\begin{proof}
The expected buyer surplus is the expected difference between the top bid and the reserve price if the second bid is below the reserve price, and the expected difference between the top two bids otherwise. Using $t$ for the top bid and $s$ for the second-highest bid, this is
\be n \int_{t=r}^{\infty} (t-r) F(r)^{n-1} f(t) \, dt \ee \label{ebs1}
\be + n (n-1) \int_{s=r}^{\infty} F(s)^{n-2} f(s) \, ds \, \int_{t = s}^{\infty} (t-s) f(t) \, dt. \ee \label{ebs2}

The first integral is
$$ = n F(r)^{n-1} \int_{t=r}^{\infty} (t-r) f(t) \, dt $$
$$ = n (1 - e^{-\lambda r})^{n-1} \int_{t=r}^{\infty} (t-r) \lambda e^{-\lambda t} \, dt.$$
Change the variable of integration from $t$ to $t+r$:
\be \int_{t=r}^{\infty} (t-r) \lambda e^{-\lambda t} \, dt \ee \label{ebsin}
$$ = \int_{t = 0}^{\infty} t \lambda e^{-\lambda (t + r)} \, dt $$
$$ = e^{-\lambda r} \int_{t = 0}^{\infty} t \lambda e^{-\lambda t} \, dt. $$
The integral is the mean of the exponential distribution, so it is $\frac{1}{\lambda}$, and Expression 1 is
\be = \frac{n}{\lambda} (1 - e^{-\lambda r})^{n-1} e^{-\lambda r}. \ee \label{ebs1r}

For Expression 2, first examine the inner integral. It is Expression 3, with $s$ instead of $r$. So it is
$$ = e^{- \lambda s} \frac{1}{\lambda}. $$
So Expression 2 is
$$ = n (n-1) \int_{s=r}^{\infty} F(s)^{n-2} \frac{1}{\lambda} f(s) \, ds \,e^{- \lambda s} $$
$$ = n (n-1) \int_{s=r}^{\infty} (1 - e^{-\lambda s})^{n-2} e^{- \lambda s} \, ds \, e^{- \lambda s}. $$
Use integration by parts ($\int u \, dv = u v - \int v \, du$), with $u = e^{- \lambda s}$ and $dv = (1 - e^{-\lambda s})^{n-2} e^{- \lambda s} \, ds$, so $du = - \lambda e^{- \lambda s} \, ds$ and $v = \frac{1}{(n-1) \lambda} (1 - e^{-\lambda s})^{n-1}$:
$$ = \left. n e^{- \lambda s} \frac{1}{\lambda} (1 - e^{-\lambda s})^{n-1} \right|_{s = r}^{\infty} $$
$$ + n \int_{s=r}^{\infty}  (1 - e^{-\lambda s})^{n-1} e^{-\lambda s} \, ds$$
$$ = \left. n e^{- \lambda s} \frac{1}{\lambda} (1 - e^{-\lambda s})^{n-1} \right|_{s = r}^{\infty} $$
$$ + \left. \frac{1}{\lambda} (1 - e^{-\lambda s})^{n} \right|_{s = r}^{\infty}.$$
Recall that $r = \frac{1}{\lambda}$:
$$ = [0 - \frac{n}{\lambda} (1 - e^{-1})^{n-1} e^{-1}] + [\frac{1}{\lambda} - \frac{1}{\lambda} (1 - e^{-1})^{n}] $$
$$ = \frac{1}{\lambda} [1 - (1 - e^{-1})^{n}] - \frac{n}{\lambda} (1 - e^{-1})^{n-1} e^{-1}. $$
The right term is the same as Expression 4, but with a minus sign. So they cancel, giving expected buyer surplus:
$$  \frac{1}{\lambda} [1 - (1 - e^{-1})^{n}]. $$
\end{proof}

For comparison, without a reserve price, the expected seller revenue is $\frac{1}{\lambda}$. 

\subsection{Reserves for Pareto Type I}
Recall that Type I Pareto distributions have cdf:
\[ F(x) = \left\{ \begin{array}{ll}
1 -  \left(\frac{a}{x}\right)^{v} & \mbox{if $x \geq a$}\\
0 & \mbox{if $x < a$} \end{array} \right. , \]
where $a > 0$ and $v > 1$.

\subsubsection{Single Bidder}
For a single bidder and reserve price $r$, expected seller revenue is $ r [1 - F(r)]$: the seller receives $r$ if the bid is at or above $r$ and zero otherwise. (We assume the bidder bids their valuation, because truth-telling is a Nash equilibrium for the second price auction with reserve prices. \cite{vickrey61,myerson81,riley81}) For a bid drawn from a Pareto Type I distribution, if $r \leq a$ then $F(r) = 0$, so expected seller revenue is $r$, indicating that we should set the reserve to at least $a$. 

For $r \geq a$,
$$ r [1 - F(r)] = \frac{a^v}{r^{v-1}}. $$
The derivative of expected seller revenue with respect to $r$ is negative for $r \geq a$:
$$ [\frac{a^v}{r^{v-1}}]' = -(v-1) \frac{a^v}{r^v}.$$
So $a$ is the optimal reserve price for a single bidder.

Note that
$$ \lim_{v \rightarrow 1_{+}} -(v-1) \frac{a^v}{r^v} = 0.$$
So as $v \rightarrow 1_{+}$, any reserve price becomes almost as effective as $r = a$. 

\subsubsection{Multiple Bidders}
But is it possible that a higher reserve price would generate higher expected seller revenue with more bidders? We will show that the answer is no. The expected seller revenue with reserve price $r$ for $n > 1$ bidders is:
$$ h(r) = n \int_{r}^{\infty} r F(r)^{n-1} f(x) \, dx $$
$$ + n (n-1) \int_{r}^{\infty} x F(x)^{n-2} f(x) [1 - F(x)] \, dx. $$
The first term is revenue from the reserve price: $n$ possible top bidders, with revenue $r$ if the other $n-1$ bids are below $r$ (probability $F(r)^{n-1}$) and the top bid is above $r$ (probability $f(x) \, dx$ integrated over $x$ from $r$ to $\infty$). The second term is revenue from the second price: $n (n-1)$ possible top and second bidder combinations, with revenue the second price $x$ if the other $n-2$ bids are below $x$ (probability $F(x)^{n-2}$) and the top bid is above $x$ (probability $[1 - F(x)]$). 

Following the derivation of optimal reserve prices \cite{myerson81,riley81}, we will find $h'(r)$, the derivative of expected seller revenue $h$ with respect to reserve price $r$. If it is negative for all $r > a$, then setting any reserve price above the minimum of the distribution decreases expected seller revenue. Pull constants out of the first integral and reverse the direction of integration of the second:
$$ h(r) = n r F(r)^{n-1} \int_{r}^{\infty} f(x) \, dx $$
$$ - n (n-1) \int_{\infty}^{r} x F(x)^{n-2} f(x) [1 - F(x)] \, dx. $$
To take the derivative, apply the product rule to the first term to get the first three terms below and the fundamental theorem of calculus to the second term to get the last term below:
$$ h'(r) = n F(r)^{n-1} \int_{r}^{\infty} f(x) \, dx $$
$$ + n r (n-1) F(r)^{n-2} f(r) \int_{r}^{\infty} f(x) \, dx $$
$$ + n r F(r)^{n-1} \left[\int_{r}^{\infty} f(x) \, dx\right]' $$
$$ - n (n-1) r F(r)^{n-2} f(r) [1 - F(r)]. $$
Note that
$$ \int_{r}^{\infty} f(x) \, dx = 1 - F(r), $$
and reverse the direction of integration then apply the fundamental theorem of calculus to the derivative of the integral in the third term:
$$ h'(r) = n F(r)^{n-1} [1 - F(r)] $$
$$ + n r (n-1) F(r)^{n-2} f(r) [1 - F(r)] $$
$$ - n r F(r)^{n-1} f(r) $$
$$ - n (n-1) r F(r)^{n-2} f(r) [1 - F(r)]. $$
The second and fourth term cancel. Combine the other two terms:
$$ h'(r) = n F(r)^{n-1} [1 - F(r) - r f(r)]. $$
This is a general expression for the derivative of expected seller revenue with respect to reserve price. Setting it equal to zero and solving for $r$ gives the well-known expression for the optimal reserve price:
$$ r^* = \frac{1 - F(r)}{f(r)}. $$

\subsubsection{Reserve Price Harms Revenue for Pareto}
For the Pareto Type I distribution, the pdf is
$$ f(x) = v \left(\frac{a}{x}\right)^{v} \frac{1}{x}. $$
So
$$ h'(r) = n F(r)^{n-1} [1 - F(r) - r v \left(\frac{a}{r}\right)^{v} \frac{1}{r}] $$
$$ = n F(r)^{n-1} [1 - F(r) - v \left(\frac{a}{r}\right)^{v}]. $$
Since 
$$ 1 - F(r) = \left(\frac{a}{r}\right)^{v}, $$
we have 
$$ h'(r) = n F(r)^{n-1} [1 - F(r)] (1 - v). $$
For $v > 1$, this is negative for all $r > a$, so no reserve price benefits the seller in expectation. 

\subsubsection{Reserves and Competition}
Does adding bidders make the expected losses from a reserve price larger or smaller? Take the ratio of $h'(r)$ for $n+1$ and $n$:
$$ \frac{(n+1) F(r)^{n} [1 - F(r)] (1 - v)}{n F(r)^{n-1} [1 - F(r)] (1 - v)} = \frac{n+1}{n} F(r). $$
If this ratio is greater than one, 
$$ \frac{n+1}{n} F(r) \geq 1, $$
then adding a bidder makes the expected revenue loss worse for the seller. Using the definition of $F$ and solving for $r$ gives: 
$$ r \leq \frac{(n+1)^{\frac{1}{v}}}{a}. $$
Equivalently, solving for n,
$$ n \geq \left(\frac{r}{a}\right)^v - 1.$$
So, for $v > 1$, adding bidders initially decreases the loss in expected revenue from having a reserve price, then beyond about $\left(\frac{r}{a}\right)^v$ bidders, adding more bidders increases the gap between no-reserve expected revenue and expected revenue with a reserve price. In short, as competition increases, the expected harm to the seller from having a reserve price also increases. 

\bibliographystyle{ACM-Reference-Format}
\bibliography{bax} 

\end{document}